\documentclass[letter,10pt, conference]{IEEEconf}
\IEEEoverridecommandlockouts
\usepackage{cite}
\usepackage{amsmath,amssymb,amsfonts}
\usepackage{bbm}
\usepackage{siunitx}
\usepackage[ruled]{algorithm}
\usepackage[noend]{algpseudocode}
\usepackage{graphicx}
\usepackage{textcomp}
\usepackage{xcolor}
\usepackage[section]{placeins}
\usepackage{flushend}
\usepackage{caption}
\usepackage{subcaption}

\def\BibTeX{{\rm B\kern-.05em{\sc i\kern-.025em b}\kern-.08em
    T\kern-.1667em\lower.7ex\hbox{E}\kern-.125emX}}

\newtheorem{lemma}{Lemma}
\newtheorem{theorem}{Theorem}
\newtheorem{remark}{Remark}
\newtheorem{proposition}{Proposition}

\DeclareTextCommand{\DJ}{OT1}{%
  \raisebox{-0.1ex}{\scalebox{0.75}[1.4]{--}}\kern-.4em D%
}

\newcommand{\algorithmicinput}{\textbf{Input:}}

\def\Input{\item[\algorithmicinput]}

\newcommand{\vx}{\boldsymbol{x}}
\newcommand{\vu}{\boldsymbol{u}}

\newcommand{\vy}{\boldsymbol{y}}

\newcommand{\xe}{\widehat{\vx}}

\newcommand{\bea}{\begin{equation} \begin{aligned}}
\newcommand{\eea}{ \end{aligned}\end{equation}}

\newcommand{\beas}{\begin{equation*} \begin{aligned}}
\newcommand{\eeas}{ \end{aligned}\end{equation*}}

\newcommand{\beq}{\begin{equation}}
\newcommand{\eeq}{\end{equation}}
\newcommand{\beqs}{\begin{equation*}}
\newcommand{\eeqs}{\end{equation*}}

\newcommand{\bpm}{\begin{pmatrix}}
\newcommand{\epm}{\end{pmatrix}}

\title{\LARGE \bf
Two-Channel Extended Kalman Filtering\\ with Intermittent Measurements*
}

\author{Vicu-Mihalis Maer$^{1}$, Zs\'ofia Lendek$^{1}$, \c{S}tefan P\^{i}rje$^{1}$, Domagoj Tolić$^{2}$,\\ Antun \DJ ura\v{s}$^{3}$, Vicko Prka\v{c}in$^{3}$, Ivana Palunko$^{3}$, and Lucian Bușoniu$^{1}$% <-this % stops a space
\thanks{*This work has been financially supported from H2020 SeaClear, a project that received funding from the European Union's Horizon 2020 research and innovation programme under grant agreement No 871295; and by the Romanian National Authority for Scientific Research, CNCS-UEFISCDI, SeaClear support project number PN-III-P3-3.6-H2020-2020-0060.}% <-this % stops a space
\thanks{$^{1}$Vicu-Mihalis Maer, Zs\'ofia Lendek, \c{S}tefan P\^{i}rje, and Lucian Bușoniu are with the Technical University of Cluj-Napoca, Cluj-Napoca, Romania
        {\tt\small \{vicu.maer, zsofia.lendek, stefan.pirje\}@aut.utcluj.ro, lucian@busoniu.net}}%
\thanks{$^{2}$Domagoj Tolić is with RIT Croatia, Dubrovnik, Croatia
        {\tt\small domagoj.tolic@croatia.rit.edu}}%
\thanks{$^{2}$Antun \DJ ura\v{s}, Vicko Prka\v{c}in, and Ivana Palunko are with the University of Dubrovnik, Dubrovnik, Croatia
        {\tt\small \{antun.djuras, vicko.prkacin, ivana.palunko\}@unidu.hr}}%
}

\begin{document}

\maketitle
\thispagestyle{empty}
\pagestyle{empty}

% \title{Two-Channel Extended Kalman Filtering with Intermittent Measurements\\
% \thanks{This work has been financially supported from H2020 SeaClear, a project that received funding from the European Union's Horizon 2020 research and innovation programme under grant agreement No 871295; and by the Romanian National Authority for Scientific Research, CNCS-UEFISCDI, SeaClear support project number PN-III-P3-3.6-H2020-2020-0060.}
% }

% \author{
% \IEEEauthorblockN{Vicu-Mihalis Maer, Zs\'ofia Lendek, \c{S}tefan P\^{i}rje}
% \IEEEauthorblockA{\textit{Technical University of Cluj-Napoca}\\
% Cluj-Napoca, Romania \\
% \{vicu.maer, zsofia.lendek, stefan.pirje\}@aut.utcluj.ro}
% \and
% \IEEEauthorblockN{Domagoj Tolić}
% \IEEEauthorblockA{
% \textit{RIT Croatia}\\%TODO check this
% Dubrovnik, Croatia \\
% domagoj.tolic@croatia.rit.edu}
% \linebreakand
% \IEEEauthorblockN{Antun \DJ ura\v{s}, Vicko Prka\v{c}in, Ivana Palunko}
% \IEEEauthorblockA{\textit{University of Dubrovnik}\\
% Dubrovnik, Croatia \\
% \{antun.djuras, vicko.prkacin, ivana.palunko\}@unidu.hr}
% \and
% \IEEEauthorblockN{Lucian Bușoniu}
% \IEEEauthorblockA{\textit{Technical University of Cluj-Napoca}\\
% Cluj-Napoca, Romania \\
% lucian.busoniu@aut.utcluj.ro}
% }

% \maketitle

\begin{abstract}
We consider two nonlinear state estimation problems in a setting where an extended Kalman filter receives measurements from two sets of sensors via two channels (2C). In the \emph{stochastic-2C} problem, the channels drop measurements stochastically, whereas in \emph{2C scheduling}, the estimator chooses when to read each channel. In the first problem, we generalize linear-case 2C analysis to obtain -- for a given pair of channel arrival rates -- boundedness conditions for the trace of the error covariance, as well as a worst-case upper bound. For scheduling, an optimization problem is solved to find arrival rates that balance low channel usage with low trace bounds, and channels are read deterministically with the expected periods corresponding to these arrival rates. We validate both solutions in simulations for linear and nonlinear dynamics; as well as in a real experiment with an underwater robot whose position is being intermittently found in a UAV camera image.
\end{abstract}

% \begin{IEEEkeywords}
% extended Kalman filter, intermittent measurements, transmission channel scheduling, underwater robot pose estimation
% \end{IEEEkeywords}

\section{Introduction}

The notion of intermittent information, whether an intrinsic or human-imposed control system property, has been extensively investigated for over two decades \cite{kjastrom2002cdc,mdlemmon2010,dtolic2017book}. These efforts naturally fall within the scope of networked control systems \cite{IEEEhowto:jhespanha,abemporad2010}. % deleted NCS acronym we have too many anyway
For example, lossy communication channels with limited bandwidths, scheduling protocols, packet collisions, sensor occlusions and limited communication/sensing ranges give rise to intermittent information. On the other hand, feedback is expected to supply estimators and controllers with up-to-date data regarding the process of interest. To resolve this tension, it is important to establish conditions leading to a satisfactory estimation and/or control performance in the presence of intermittent information. Herein, we focus on estimation under intermittent measurements.

In particular, we consider a scenario in which a nonlinear system is observed by two sets of sensors via two respective channels. Either set of sensors may in general be local, but we use the name ``channel'' even in that case, since it is standard \cite{Liu04}. Our main objective is to propose solutions for two related problems in this nonlinear two-channel (2C) setting: \emph{stochastic-2C estimation}, where the two channels drop measurements stochastically with different probabilities (an intrinsic property), and \emph{2C scheduling}, in which the estimation algorithm may choose whether to use either of these channels at each discrete time step (a human-imposed property).

The scenario above is motivated by a practical problem that occurs in the European Horizon 2020 SeaClear project \cite{seaclear}. An Unmanned Underwater Vehicle (UUV) has access to its internal sensors at every step, but these sensors cannot compensate for drift in the position estimate. On the other hand, underwater absolute position sensors are expensive and sometimes unreliable, so instead, the UUV position is determined using the camera image of an Unmanned Aerial Vehicle (UAV), which is however possible only when the UUV is close enough to the surface to be visible. Thus, this second sensor is available intermittently, and the UUV chooses when to resurface to make it available. Beyond this specific case, the two-channel (2C) scenario appears e.g. when shared communication networks with limited throughput are encountered \cite{abemporad2010,dtolic2017book}.
% (e.g., packet collisions ought to be avoided and/or the bandwidth does not allow for both observations to be sent simultaneously). 

% In general, both sensors can be off board the UUV so that their measurements are communicated to the UUV. Hence, the term \textit{two channels}, not two sensors, is utilized more frequently .
%\textbf{TODO}  Domagoj, is this a good name given the NCS literature? $[$Yes, keep it as it is.$]$

% Our objective in this work is to propose a solution to this \emph{2C sensor scheduling} problem for nonlinear processes. 

In linear Kalman filtering with intermittent measurements on a single channel, conditions on the boundedness of estimation error covariance were developed in \cite{Sinopoli04}. That reference showed that there exists a critical value for the arrival rate of the single-channel measurements, beyond which the covariance becomes unbounded. This critical probability has been further analyzed in \cite{Plarre09}. Reference \cite{Liu04} extended the results in \cite{Sinopoli04} to stochastic-2C Kalman filtering, with probabilities $\lambda_1$ and $\lambda_2$ of successfully delivering  measurements. The authors of \cite{Liu04}  proved the existence of a sharp transition curve for the stability of the iteration on the covariance matrices and show that, when one of the arrival probabilities is fixed, the critical value of the other one can be found by solving a series of linear matrix inequality (LMI) feasibility problems. In the previous works, the arrival probabilities were assumed to be i.i.d.~from a Bernoulli distribution. The case when the observations become available according to a Markov process modelling a Gilbert-Elliot channel has been considered in \cite{Mo12}. 

As a first contribution, we ``turn around'' the estimation method from \cite{Liu04} so as to apply it to \emph{linear 2C scheduling}. To this end, we read each channel $i \in \{1, 2\}$ with period $T_i = \lfloor \frac{1}{\lambda_i} \rfloor$,\footnote{Operator $\lfloor \cdot \rfloor$ takes the floor of the argument.} which ensures that the guarantees of \cite{Liu04} apply with the corresponding values of $\lambda_i$. To find a pair of arrival rates $(\lambda_1, \lambda_2)$, we optimize over a predefined set of candidate pairs from which we exclude infeasible values that lead to unstable estimates. The objective function balances low values of $\lambda_i$, so as to reduce channel usage, with a low trace of the error covariance matrix, so as to improve estimation accuracy.

% \textbf{Namely, the analysis of \cite{Liu04} applies to any sensor schedule realization, and reading channel $i$ after the expected number of steps $T_i$ corresponding to $\lambda_i$ -- one such realization -- what does this mean?}. \textbf{Somebody should fix the previous sentence. Even before it was too long and hard to read.}\textbf{TODO} by the way Mihalis, how do we handle the case where the expected period is noninteger?

% \textbf{Trying to make sense of para above, please rewrite}
% As a first contribution, we adapt the stochastic estimation method from \cite{Liu04} to our 2C \emph{scheduling} problem. 

% \textbf{This paragraph now reads well, but Mihalis might want to add additional details from the old paragraph above.}

Our larger objective is however, as stated, to devise a solution for the nonlinear case, and for that purpose, we first analyze \emph{stochastic-2C} Extended Kalman Filtering (EKF) for a class of discrete-time nonlinear systems in which the linearized transition dynamics vary in a polytope. For a given pair of arrival rates, we develop LMI conditions to establish boundedness of the covariance matrices and compute a worst-case upper bound. To our best knowledge, the present paper is the first to consider the nonlinear 2C setting. Stochastic stability of the discrete-time EKF has been investigated in~\cite{Reif99}, and the case with intermittent measurements on a \emph{single} channel has been analyzed in \cite{Kluge10, Wang13, Liu17}. While the discrete-time model considered in \cite{Kluge10} is quite general, it has the shortcoming that the measurement matrix, albeit time-varying, must be invertible. Specific variants of EKF with intermittent measurements have been developed for localization \cite{Ahmad13} and tracking \cite{Hicks12, Chen15, Shi16}. Stability of the unscented Kalman filter with intermittent observations has been analyzed in \cite{Li12}. %Regarding \cite{Liu17}, the conditions require that the initial estimation error satisfy a condition related to the bounds on the system matrices. 
%We assume that the state matrices vary in a polytope and consider a scenario similar to \cite{Liu04, Sinopoli04} to analyze the statistical properties of the error covariance matrices of the EKF used together with intermittent observations on two channels. When arrival rates are given, 
 %\textbf{Somewhere we should compare our work with papers investigating EKF with intermittent observations. Do they readily boil down to 2C case we investigate here? Some of those papers I have sent in an e-mail several months ago...}

Moreover, we consider \emph{2C scheduling} for the EKF, where we solve a similar optimization problem to the one from the linear case in order to get~$\lambda_i$, and then read the channels with the corresponding periods $T_i$. Differently from the linear case, we apply the newly developed EKF conditions. Since this solution may sometimes be conservative, we additionally propose an empirical, iterative application of the KF conditions, which assumes that the nonlinear dynamics are slowly varying. This last approach linearizes the nonlinear system around the current operating point and recomputes $\lambda_i$ by solving the linear-case optimization problem. The procedure is repeated when the dynamics deviate significantly from the previous linearization. 

% Namely, we assume that the dynamics, although nonlinear, are slowly varying, such that, once an arrival rate has been considered, the dynamics do not change significantly before the arrival of the next measurement. Once a measurement becomes available and the estimate has been updated, another linear approximation is computed, together with the corresponding arrival rate that ensures

%\[\alpha(\lambda_1,\lambda_2) = tr(P_{\lambda_1,\lambda_2}) + e^{\frac{1}{1-\lambda_1}} + e^{\frac{1}{1-\lambda_2}}\]where $P_{\lambda_1,\lambda_2}$
% \textbf{TODO} @Mihalis along the way in the previous 3 paras we also need to explain how we impose conditions on the trace, optimization objective for scheduling etc. \textbf{Yes, this should be in the intro, not that integer discussion (provided I got it correctly).}

To illustrate the approaches developed, we start with simulations in the linear KF case, since the approach in \cite{Liu04} was not validated numerically. Then, we apply 2C scheduling to estimate the pose of a nonlinear rigid-body system using an EKF. Finally, we present a real-life underwater robotics experiment where the onboard-channel is always on, and the UAV-camera-based positioning channel is read with our 2C-scheduling approach. In this experiment, the state estimate is validated against underwater acoustic positioning. %\textbf{Drones could be underwater vehicles as well, even though people almost always use them for aerial vehicles. Hence, I prefer UAV over drone.}

Next, in Section \ref{sec:analysis}, we provide the analysis of the stochastic 2C EKF, followed by the methods for 2C scheduling in both the linear and nonlinear cases in Section \ref{sec:scheduling}. Simulation and real-robot experimental results are given in Sections \ref{sec:sim} and \ref{sec:exp}, respectively. Section \ref{sec:conclusions} concludes the paper.

% \section{Stochastic-2C filtering: Linear, KF background}

% \textbf{TODO} Mihalis or Zsofi

% review of the original result, with:
%equations
%LMI
%traceLMI

\section{Analysis of Stochastic-2C Extended Kalman Filtering}
\label{sec:analysis}

The main theoretical contribution of the paper is to analyze the statistical properties of the covariance matrices for the EKF, when applied to a class of discrete-time nonlinear systems with intermittent 2C measurements.

We consider the discrete-time nonlinear system
\bea
\label{eq:full}
\vx_{k+1}&=f(\vx_k)+B\vu_{k}+w_{k},\\
\vy_{k}&=C\vx_{k}+v_{k},
\eea
where $\vx_{k}\in \mathcal{R}^{n_x}$ denotes the state at time $k$, $\vu\in \mathcal{R}^{n_u}$ is the input, $\vy \in \mathcal{R}^{n_y}$ is the measured output, and $w$ and $v$ are zero-mean white noises, with covariances $Q=Q^T> 0$ and $R=R^T>0$, respectively. $B$ is the input matrix and $f:\mathcal{R}^n\rightarrow \mathcal{R}^n$ is a vector function. Note that a linear input dependence is assumed.

We consider a scenario similar to \cite{Liu04}, with the measurement vector $\vy$ supplied by two sets of sensors, whose outputs are encoded separately and sent via two different channels. The output $\vy$ is consequently partitioned as
\bea
\bpm \vy_{k,1}\\ \vy_{k,2}\epm=\bpm C_1\\C_2 \epm \vx_{k}+\bpm v_{k,1}\\v_{k,2}\epm
\eea
with $\vy_{k}=\bpm \vy_{k,1}\\ \vy_{k,2}\epm$, $C=\bpm C_1\\C_2 \epm$, $v_{k}=\bpm v_{k,1}\\v_{k,2} \epm$. The measurement noises are $v_{1,k} \sim N(0,R_{11})$ and $v_{2,k} \sim N(0,R_{22})$, where $R=\bpm R_{11} &R_{12}\\R_{21}& R_{22}\epm $. The channels may be lossy and not all measurement packages are received. The arrival of measurement $\vy_{k,i}$, $i=1,\,2,$ at time $k$ is given by a binary variable, $\gamma_{k,i}$, sampled from a Bernoulli process with probability $P(\gamma_{k,i}=1)=\lambda_i$, $i=1,\,2$. We consider independent sensors and channels, so the probability of both measurements arriving at the same time is $\lambda_1 \lambda_2$.

To estimate the states of the system, we consider an EKF. We denote $A_{k}:=\frac{\partial f}{\partial \vx}|_{\xe_{k|k}}$ and develop general results. Later on, to obtain conditions that are easier to implement, we will consider the case when $A_{k} \in \textrm{Co}(\mathcal{A}_j)$, $j=1,\,2,\,\cdots,\,a$, where $\textrm{Co}(.)$ denotes the convex hull.

The time update is independent of the measurements and the predictions are based on model~\eqref{eq:full},~i.e.,
\bea
\label{eq:pred}
\xe_{k+1|k}&=f(\xe_{k|k})+B\vu_{k},\\
P_{k+1|k}&=A_{k}P_{k|k}A_{k}^T+Q,
\eea
where the usual notations are used, i.e., $k+1|k$ denotes prediction and so $\xe_{k+1|k}$ is the predicted state, $\xe_{k|k}$ is the estimated state after the $k$th measurement has been processed, $P_{k|k}$ is the estimator covariance matrix at the same moment, and so on.

Following \cite{Sinopoli04, Liu04}, the update equations depend on the measurements received. If no measurements are received, then only prediction is performed. Otherwise, the state and the covariance matrices are updated using the received measurements, i.e.,
\bea
\label{eq:state_cov_update}
&\xe_{k+1|k+1}=\xe_{k+1|k}+\\
&+\gamma_{k+1,1}\gamma_{k+1,2}P_{k+1|k}C^T(CP_{k+1|k}C^T+R)^{-1}\times\\
&\times(\vy_{k+1}-C\xe_{k+1|k})\\
&+\gamma_{k+1,1}(1-\gamma_{k+1,2})P_{k+1|k}C_1^T(C_1P_{k+1|k}C_1^T+R_{11})^{-1}\times\\
&\times(\vy_{k+1,1}-C_1\xe_{k+1|k})\\
&+(1-\gamma_{k+1,1})\gamma_{k+1,2}P_{k+1|k}C_2^T(C_2P_{k+1|k}C_2^T+R_{22})^{-1}\times\\
&\times(\vy_{k+1,2}-C_2\xe_{k+1|k}),\\
&P_{k+1|k+1}=P_{k+1|k}-\\
&-\gamma_{k+1,1}\gamma_{k+1,2}P_{k+1|k}C^T(CP_{k+1|k}C^T+R)^{-1}CP_{k+1|k}\\
&-\gamma_{k+1,1}(1-\gamma_{k+1,2})P_{k+1|k}C_1^T(C_1P_{k+1|k}C_1^T+R_{11})^{-1}\times\\ &\times C_1P_{k+1|k}\\
&-(1-\gamma_{k+1,1})\gamma_{k+1,2}P_{k+1|k}C_2^T(C_2P_{k+1|k}C_2^T+R_{22})^{-1}\times\\ &\times C_2P_{k+1|k}.
\eea
In what follows, we use the simplified notation $P_{k+1} := P_{k+1|k}$. Then, the predicted covariance matrix at sample $k+1$ can be expressed as
\bea
\label{eq:cov_pred}
&P_{k+1}=A_kP_{k}A_k^T+Q-\\
&-\gamma_{k,1}\gamma_{k,2}A_kP_{k}C^T(CP_{k}C^T+R)^{-1}CP_{k}A_k^T\\
&-\gamma_{k,1}(1-\gamma_{k,2})A_kP_{k}C_1^T(C_1P_{k}C_1^T+R_{11})^{-1}C_1P_{k}A_k^T\\
&-(1-\gamma_{k,1})\gamma_{k,2}A_kP_{k}C_2^T(C_2P_{k}C_2^T+R_{22})^{-1}C_2P_{k}A_k^T.
\eea

\begin{remark} Note that contrary to the linear case, the matrices $P_{k}$, $P_{k+1|k}$ are not the error covariance matrices. However, for simplicity, we will refer to them as such. Furthermore, since measurements may be lost, both $\xe$ and $P$ become random variables (as they depend on the random variables $\gamma_1$ and $\gamma_2$).
\end{remark}

In this setting, our goal is to determine conditions on the existence of an upper bound on the covariance matrices $P_k$, given the arrival probabilities $\lambda_1$ and $\lambda_2$ of the measurements; and to determine the minimal arrival probabilities $\lambda_1$ and $\lambda_2$ such that the covariance matrices remain bounded. %\textbf{Which matrix norm/measure are we bounding? This should be mentioned in the intro as well. -> No norms used, probably can be done like that, but not in 4 days. So no norms mentioned. OK, then just say what we are doing. Norms could be done in the future. No problem.} 
In order to do this, we exploit and generalize some of the results presented in \cite{Liu04} and \cite{Sinopoli04}. Similar to the mentioned results, we define the functions:
\bea
\label{eq:gfunc}
&g_{\lambda_1 \lambda_2}(k,X):=A_kXA_k^T+Q-\\
&-\lambda_{1}\lambda_{2}A_kXC^T(CXC^T+R)^{-1}CXA_k^T\\
&-\lambda_{1}(1-\lambda_{2})A_kXC_1^T(C_1XC_1^T+R_{11})^{-1}C_1XA_k^T\\
&-(1-\lambda_{1})\lambda_{2}A_kXC_2^T(C_2XC_2^T+R_{22})^{-1}C_2XA_k^T
\eea
and
\bea
\label{eq:phifunc}
&\phi(k,K_k,K_{k,1},K_{k,2},X):=\\
&(1-\lambda_1)(1-\lambda_2)(A_kXA_k^T+Q)+\lambda_{1}\lambda_{2}(F_kXF_k^T+V_k)\\
&+\lambda_{1}(1-\lambda_{2})(F_{k,1}XF_{k,1}^T+V_{k,1})\\
&+(1-\lambda_{1})\lambda_{2}(F_{k,2}XF_{k,2}^T+V_{k,2}),
\eea
where $F_k=A_k+KC$, $F_{k,1}=A_k+K_{k,1}C_1$, $F_{k,2}=A_k+K_{k,2}C_2$, $V_k=Q+K_kRK_k^T$, $V_{k,1}=Q+K_{k,1}R_{11}K_{k,1}^T$, $V_{k,2}=Q+K_{k,2}R_{22}K_{k,2}^T$, and $X \ge 0$. We also define $K_{k,x}=-A_kXC^T(CXC^T+R)^{-1}$, $K_{k,1,x}=-A_kXC_2^T(C_1XC_1^T+R_{11})^{-1}$, $K_{k,2,x}=-A_kXC_2^T(C_2XC_2^T+R_{22})^{-1}$.

It is straightforward to see that in the linear case, when $A_k$ is constant, the problem reduces to that in \cite{Liu04}. Furthermore, for a fixed $k$, the properties in \cite{Liu04} and Lemmas 1 and 2 in \cite{Sinopoli04} hold. Specifically, for any given $k$, we have that $E(P_{k+1}|P_k)=g_{\lambda_1 \lambda_2}(k,P_k)$, $g_{\lambda_1 \lambda_2}(k,X)$ is concave and non-decreasing in $X$, $g_{\lambda_1 \lambda_2}(k,X) \ge (1-\lambda_1)(1-\lambda_2)A_kXA_k^T+Q$, and thus it is possible to obtain a lower bound on $E(g_{\lambda_1 \lambda_2}(k,P_k))$. Furthermore, with the definitions of $K_{k,x}$, $K_{k,1,x}$, and $K_{k,2,x}$, it follows that $\phi(k,K_{k,x},K_{k,1,x},K_{k,2,x},X)=\min_{K_{k,x},K_{k,1,x},K_{k,2,x}}\phi(k,K_k,K_{k,1},K_{k,2},X)$.

Next, we state some lemmas that will be useful for developing the boundedness conditions.

\begin{lemma}
\label{lem:iter_l}
Define the operator
\bea
\label{op:l}
&\mathcal{L}(k,Y):=(1-\lambda_1)(1-\lambda_2)(A_kYA_k^T)+\lambda_1\lambda_2F_kYF_k^T\\
&+\lambda_1(1-\lambda_2)F_{k,1}YF_{k,1}^T+(1-\lambda_1)\lambda_2F_{k,2}YF_{k,2}^T.
\eea
$\mathcal{L}(k,Y)$ is linear in $Y$ and $\mathcal{L}(k,Y)\ge 0$. Assume that there exists $\bar{Y}>0$ such that $\bar{Y}> \mathcal{L}(k,\bar{Y})$, $\forall k$. Then,
\begin{enumerate}
\item $\forall W\ge 0$, $\lim_{k \rightarrow \infty}\mathcal{L}(k,W)=0$.
\item Let $U_k\ge 0$ bounded and consider $Y_{k+1}=\mathcal{L}(k,Y_k)+U_k$ initialized at $Y_0$. Then, the sequence $Y_k$ is bounded.
\end{enumerate}
\end{lemma}

\begin{proof}
Based on the assumption that $\exists \bar{Y}>0$ so that $\bar{Y} > \mathcal{L}(k,\bar{Y})$, $\forall k$, one can choose $0\le r<1$ so that $\mathcal{L}(k,\bar{Y})<r \bar{Y}$, $\forall k$. The rest of the proof follows the same lines as that of Lemma 3 in \cite{Sinopoli04}. 
\end{proof}

\begin{lemma}
\label{lem:upper}
Consider the operator $\phi(k,K,K_1,K_2,X)$ defined in \eqref{eq:phifunc}. If there exist matrices $K_k$, $K_{k,1}$, $K_{k,2}$, and $\bar{P}>0$ so that $\bar{P}>\phi(k,K_k,K_{k,1},K_{k,2}, \bar{P})$, then the sequence $P_{k+1}=g_{\lambda_1, \lambda_2}(k, P_0)$ is bounded for any $P_0$.
\end{lemma}

\begin{proof}
Define the matrices $F_{k}=A_k+K_kC$, $F_{k,1}=A_k+K_{k,1}C$, and $F_{k,2}=A_k+K_{k,2}C$ and consider the operator $\mathcal{L}(k,Y)$ defined in \eqref{op:l}. It is easy to verify that $\phi(k,K_k,K_{k,1},K_{k,2},X)=\mathcal{L}(k,X)+(1-\lambda_1)(1-\lambda_2)Q+\lambda_1\lambda_2 V_k+\lambda_1(1-\lambda_2)V_{k,1}+(1-\lambda_1)\lambda_2 V_{k,2}$, where $V_{k}=Q+K_kRK_k^T$, $V_{k,1}=Q+K_{k,1}R_{11}K_{k,1}^T$, $V_{k,2}=Q+K_{k,2}R_{22}K_{k,2}^T$, i.e., $\phi(k,K_k,K_{k,1},K_{k,2},X)=\mathcal{L}(k,X)+U_k$, with $U_k$ being defined as $U_k=(1-\lambda_1)(1-\lambda_2)Q+\lambda_1\lambda_2 V_k+\lambda_1(1-\lambda_2)V_{k,1}+(1-\lambda_1)\lambda_2 V_{k,2}$ and, since $Q>0$, $R>0$, $R_{11}>0$, $R_{22}>0$, $U_k>0$. Using the assumption that $\bar{P}>\phi(k,K_k,K_{k,1},K_{k,2}, \bar{P})$, we have $\bar{P}> \mathcal{L}(k,\bar{P})+U_k>\mathcal{L}(k,\bar{P})$.

On the other hand, we also have $P_{k+1}=g_{\lambda_1 \lambda_2}(k,P_k)\le \phi(k,K_k,K_{k,1},K_{k,2},P_k)=\mathcal{L}(k,P_k)+U_k$. Based on Lemma~\ref{lem:iter_l}, the sequence $P_k$ is bounded. 
\end{proof}

We are now ready to state the main result on the boundedness of the covariance matrices:
\begin{theorem}
\label{th:boundedness}
Consider the operator $\phi(k,K_k,K_{k,1},K_{k,2},X)$ defined in \eqref{eq:phifunc}.
%=(1-\lambda_1)(1-\lambda_2)(A_kXA_k^T+Q)+\lambda_{1}\lambda_{2}(F_kXF_k^T+V_k)+\lambda_{1}(1-\lambda_{2})(F_{k,1}XF_{k,1}^T+V_{k,1})+(1-\lambda_{1})\lambda_{2}(F_{k,2}XF_{k,2}^T+V_{k,2})$, where $F_{k}=A_k+K_kC$, $F_{k,1}=A_k+K_{k,1}C_1$, $F_{k,2}=A_k+K_{k,2}C_2$, $V_{k}=Q+K_kRK_k^T$, $V_{k,1}=Q+K_{k,1}R_{11}K_{k,1}^T$, $V_{k,2}=Q+K_{k,2}R_{22}K_{k,2}^T$.
%\textbf{Should we just say "Consider the operator $\phi(k,K_k,K_{k,1},K_{k,2},X)$ defined in (7)" since all previous definitions are already in (7) and right below it?} 
Assume that there exist matrices $K_k$, $K_{k,1}$, $K_{k,2}$ and a positive matrix $P=P^T>0$ so that $P>\phi(k,K_k,K_{k,1},K_{k,2},P)$. Then, for any initial condition $P_0\ge 0$, $\lim_{k\rightarrow\infty}P_k=\lim_{k\rightarrow \infty}g_{\lambda_1 \lambda_2}(k,P_0)$ is bounded.
\end{theorem}

\begin{remark}
Contrary to the results in \cite{Liu04}, \cite{Sinopoli04}, a single point of transition or a transition curve between boundedness and unboundedness of $P_k$ in general will not exist. However, a worst-case upper bound on the critical probabilities $\lambda_1$ and $\lambda_2$ can be computed similarly to Theorem 2 of \cite{Liu04}.
\end{remark}

If $(\lambda_1,\,\lambda_2)$ is such that boundedness is maintained, a limit on the bound of $P_k$ is given by the following theorem:
\begin{theorem}
Assume that $(A_k,Q)$ is controllable, $(A_k, C)$ is detectable $\forall k$ and the pair $(\lambda_1, \lambda_2)$ is such that $P_k$ is bounded. Then, $\lim_{k\rightarrow\infty}P_k\le V$, where $V >g_{\lambda_1 \lambda_2}(k,V)$, $V>0$, $\forall k$.
\end{theorem}
\begin{proof}
Follows the lines of Theorem 6 in \cite{Sinopoli04}, taking into account that $V>0$ and $V >g_{\lambda_1 \lambda_2}(k,V)$ have to hold~$\forall k$. 
\end{proof}

Next, we formulate a theorem to compute a worst-case upper bound on the probabilities.
\begin{theorem}
\label{th:lambdas}
If $(A_k,Q)$ is controllable, $(A_k, C)$ is detectable $\forall k$, then the following statements are equivalent:
\begin{enumerate}
\item there exists $\bar{X}>0$ such that $\bar{X} > g_{\lambda_1 \lambda_2}(k,\bar{X})$, $\forall k$,
\item there exist $K_k$, $K_{k,1}$, $K_{k,2}$ and $\bar{X}>0$ so that $\bar{X}>\phi(k,K_k,K_{k,1},K_{k,2},\bar{X})$, $\forall k$,
\item there exist $\bar{Z}_k$, $\bar{Z}_{k,1}$, $\bar{Z}_{k,2}$, and $0<\bar{Y}\le I$ such that $\forall k$
\bea
\label{eq:lmi_lam}
\Psi_k=\bpm
Y&\sqrt{\lambda_1 \lambda_2}(YA_k+Z_kC)&\Psi_{13} &\Psi_{14} \\
(*)&Y&0&0\\
(*)&(*)&Y&0\\
(*)&(*)&(*)&Y
\epm>0,
\eea
where $\Psi_{13}=\sqrt{\lambda_1 (1-\lambda_2)}(YA_k+Z_{k,1}C_1)$, $\Psi_{14}=\sqrt{\lambda_2 (1-\lambda_1)}(YA_k+Z_{k,2}C_2)$ and $(*)$ denotes the symmetric term.
\end{enumerate}
\end{theorem}

The condition formulated in \eqref{eq:lmi_lam} is bilinear, due to having to search for $\lambda_1$, $\lambda_2$, $Y$, etc. However, if either $\lambda_1$ or $\lambda_2$ is fixed, then bisection can be used to find the other probability, i.e., one will have to solve a set of LMI feasibility problems. Once a suitable pair $(\lambda_1,\,\lambda_2)$ is found, an upper bound on the matrix $P_k$ can be computed as in the next theorem.

\begin{theorem}
\label{th:cov}If there exist $\bar{Z}_k$, $\bar{Z}_{k,1}$, $\bar{Z}_{k,2}$ and $0<\bar{Y}\le I$ such that \eqref{eq:lmi_lam} is satisfied, then an upper bound on $\lim_{k \rightarrow \infty}g_{\lambda_1 \lambda_2}(k,V)$ can be found  by solving
\bea \label{eq:tracebound}
&\textrm{argmax}_V \textrm{Trace}(V)\\
&\textrm{subject to } V>0,\,\Gamma(V)\ge 0
\eea
where
\bea
\Gamma (V)=\bpm
A_kVA_k^T+Q-V&\sqrt{\lambda_1 \lambda_2}A_kVC^T&\Gamma_{13} &\Gamma_{14} \\
(*)&CVC^T+R&0&0\\
(*)&(*)&\Gamma_{33}&0\\
(*)&(*)&(*)&\Gamma_{44}
\epm
\eea
$\Gamma_{13}=\sqrt{\lambda_1 (1-\lambda_2)}A_kVC_1^T$, $\Gamma_{14}=\sqrt{\lambda_2 (1-\lambda_1)}A_kVC_2^T$, $\Gamma_{33}=C_1VC_1^T+R_{11}$, $\Gamma_{44}=C_2VC_2^T+R_{22}$.
\end{theorem}

In the developments so far we have considered that the state matrices $A_k$ vary in time without any further constraints. This means that the number of conditions to be solved for Theorems~\ref{th:lambdas} and~\ref{th:cov}, respectively, is infinite. In practice, however, a domain (possibly overestimated) in which the matrices vary can usually be determined. Therefore let us now consider the case when $A_{k} \in \textrm{Co}(\mathcal{A}_j)$, $j=1,\,2,\,\cdots,\,a$, and there exist functions $h_j(\cdot)$ such that $h_j(k) \ge 0$, $\sum_{j=1}^ah_j(k)=1$ and $A_k=\sum_{j=1}^ah_j(k)\mathcal{A}_j$, i.e., each matrix $A_k$ can be expressed as the convex combination of the matrices $\mathcal{A}_j$, $j=1,\,2,\,\cdots,\,a$. In such a case, a sufficient condition for condition \eqref{eq:lmi_lam} to hold~is 
\begin{proposition}
\label{prop:lmipoly}
If there exist $\mathcal{Z}_j$, $\mathcal{Z}_{j,1}$, $\mathcal{Z}_{j,2}$ and $0<\bar{Y}\le I$ such that $\forall k$
\bea
\label{eq:lmi_lam_poly}
\Psi^2_j=\bpm
Y&\sqrt{\lambda_1 \lambda_2}(Y\mathcal{A}_j+\mathcal{Z}_jC)&\Psi^2_{13} &\Psi^2_{14} \\
(*)&Y&0&0\\
(*)&(*)&Y&0\\
(*)&(*)&(*)&Y
\epm>0
\eea
holds for $j=1,\,2,\,\dots,\,a$, where $\Psi^2_{13}=\sqrt{\lambda_1 (1-\lambda_2)}(Y\mathcal{A}_j+\mathcal{Z}_{j,1}C_1)$, $\Psi^2_{14}=\sqrt{\lambda_2 (1-\lambda_1)}(Y\mathcal{A}_j+\mathcal{Z}_{j,2}C_2)$, then condition \eqref{eq:lmi_lam} holds.
\end{proposition}
\begin{proof}
Let $Z_k=\sum_{j=1}^ah_j(k)\mathcal{Z}_j$, $Z_{k,1}=\sum_{j=1}^ah_j(k)\mathcal{Z}_{k,j}$, and $Z_{k,2}=\sum_{j=1}^ah_j(k)\mathcal{Z}_{j,2}$. Taking into account that $A_k=\sum_{j=1}^ah_j(k)\mathcal{A}_j$, $\Psi_k=\sum_{j=1}^ah_j(k)\Psi^2_j$. Since $h_j(k) \ge 0$ and $\sum_{j=1}^ah_j(k)=1$, $\Psi^2_j>0$, $j=1,\,2,\,\dots,\,a$, that implies~$\Psi_k~>~0$. \hfill$\Box$
\end{proof}

A sufficient condition for Theorem~\ref{th:cov} is formulated as follows:
\begin{proposition}
\label{prop:cov}If there exist $\mathcal{Z}_j$, $\mathcal{Z}_{j,1}$, $\mathcal{Z}_{j,2}$, $j=1,\,2,\,\dots,\,a$, and $0<\bar{Y}\le I$ such that \eqref{eq:lmi_lam_poly} is satisfied, then an upper bound on $\lim_{k \rightarrow \infty}g_{\lambda_1 \lambda_2}(k,V)$ can be found solving
\bea
\label{eq:opt_poly}
&\textrm{argmax}_V \textrm{Trace}(V)\\
&\textrm{subject to } V>0,\,\Gamma^2(V)\ge 0,\, j=1,\,2,\,\dots,\,a
\eea
where
\bea
\label{eq:cov_poly}
&\Gamma^2 (V)=\\
&\bpm
\mathcal{A}_j+\mathcal{A}_j^T+Q-V&\sqrt{\lambda_1 \lambda_2}\mathcal{A}_jVC^T&\Gamma^2_{13} &\Gamma^2_{14}&I \\
(*)&CVC^T+R&0&0&0\\
(*)&(*)&\Gamma^2_{33}&0&0\\
(*)&(*)&(*)&\Gamma^2_{44}&0\\
(*)&(*)&(*)&(*)&V
\epm
\eea
$\Gamma^2_{13}=\sqrt{\lambda_1 (1-\lambda_2)}\mathcal{A}_jVC_1^T$, $\Gamma^2_{14}=\sqrt{\lambda_2 (1-\lambda_1)}\mathcal{A}_jVC_2^T$, $\Gamma^2_{33}=C_1VC_1^T+R_{11}$, $\Gamma^2_{44}=C_2VC_2^T+R_{22}$.
\end{proposition}
\begin{proof}
Recall that $\forall k$, $A_k=\sum_{j=1}^ah_j(k)\mathcal{A}_j$. Note that $A_kVA_k^T\ge A_k+A_k^T-V^{-1}$, thus $\bpm
A_kVA_k^T+Q-V&\sqrt{\lambda_1 \lambda_2}A_kVC^T&\Gamma_{13} &\Gamma_{14} \\
(*)&CVC^T+R&0&0\\
(*)&(*)&\Gamma_{33}&0\\
(*)&(*)&(*)&\Gamma_{44}
\epm \ge \bpm
A_k+A_k^T+Q-V-V^{-1}&\sqrt{\lambda_1 \lambda_2}A_kVC^T&\Gamma_{13} &\Gamma_{14} \\
(*)&CVC^T+R&0&0\\
(*)&(*)&\Gamma_{33}&0\\
(*)&(*)&(*)&\Gamma_{44}
\epm$. Applying a Schur complement on the element $A_k+A_k^T+Q-V-V^{-1}$ and taking into account that $A_k=\sum_{j=1}^ah_j(k)\mathcal{A}_j$, $\forall k$, we obtain \eqref{eq:cov_poly}. \hfill$\Box$
\end{proof}

\begin{remark} % isn't this too negatively stated?
The conditions stated in Propositions~\ref{prop:lmipoly} and \ref{prop:cov} may in some cases be overly conservative as, if satisfied, they will guarantee boundedness and determine an upper bound, respectively, for all the nonlinear systems in the polytope.  
\end{remark}

\begin{remark}
It can easily be seen that the conditions in Propositions \ref{prop:lmipoly} and \ref{prop:cov} reduce to those in Theorems 5 and 6 in~\cite{Liu04} in the case of a constant state matrix, i.e., when $A_k=A$, $\forall k$ or $\mathcal{A}_j=A$, $j=1,\,2,\,\dots,\,a$.  Specifically, condition \eqref{eq:lmi_lam_poly} reduces to 
\bea
\label{eq:lmi_lin}
\Psi^l=\bpm
Y&\sqrt{\lambda_1 \lambda_2}(YA+ZC)&\Psi^l_{13} &\Psi^l_{14} \\
(*)&Y&0&0\\
(*)&(*)&Y&0\\
(*)&(*)&(*)&Y
\epm>0
\eea
with $\Psi^l_{13}=\sqrt{\lambda_1 (1-\lambda_2)}(YA+Z_1C_1)$, $\Psi^l_{14}=\sqrt{\lambda_2 (1-\lambda_1)}(YA+Z_2C_2)$ and the bound on the covariance can be calculated by solving 
\bea
\label{eq:opt_V}
&\textrm{argmax}_{V} \textrm{Trace}(V)\\
&\textrm{subject to } V>0,\,\Gamma^l(V)\ge 0,\, j=1,\,2,\,\dots,\,a
\eea
where
\bea
\label{eq:lin_V}
&\Gamma^l(V)=\\
&\bpm
AVA^T+Q&\sqrt{\lambda_1 \lambda_2}AVC^T&\Gamma^l_{13} &\Gamma^l_{14} \\
(*)&CVC^T+R&0&0\\
(*)&(*)&\Gamma^2_{33}&0\\
(*)&(*)&(*)&\Gamma^l_{44}
\epm
\eea
$\Gamma^l_{13}=\sqrt{\lambda_1 (1-\lambda_2)}AVC_1^T$, $\Gamma^l_{14}=\sqrt{\lambda_2 (1-\lambda_1)}AVC_2^T$, $\Gamma^l_{33}=C_1VC_1^T+R_{11}$, $\Gamma^2_{44}=C_2VC_2^T+R_{22}$.
Furthermore, if only a single intermittent-measurement channel is considered, the conditions become those in Theorems 5 and 6 in \cite{Sinopoli04}.
\end{remark}
%implementation

%basik ekf transition
\section{Two-Channel Scheduling}
\label{sec:scheduling}

Moving now to the 2C scheduling problem, a straightforward way to solve it is to find a pair of arrival rates $(\lambda_1, \lambda_2)$ that satisfy \eqref{eq:lmi_lam_poly} and therefore ensure boundedness of the error covariance matrix; and then to read each channel $i \in \{1, 2\}$ with period $T_i = \lfloor \frac{1}{\lambda_i} \rfloor$. Since there may be many such pairs, we formulate an optimization problem to select the best one:
\bea
\label{eq:l1l2objective}
\min_{(\lambda_1, \lambda_2) \in L \text{ s.t. \eqref{eq:lmi_lam_poly}}}
                \tau + e^{\frac{1}{1-\lambda_1}} + e^{\frac{1}{1-\lambda_2}},
\eea
where set $L$ contains finitely many candidate pairs $(\lambda_1, \lambda_2)$, and $\tau=\mathrm{Trace}(V^*)$ with $V^*$ being the solution of \eqref{eq:opt_poly} for the pair $(\lambda_1,\,\lambda_2)$. We take finitely many pairs to be able to solve the optimization problem by enumeration. The objective function \eqref{eq:l1l2objective} aims to minimize both the value of the trace (via its upper bound as a proxy) and the arrival rates (channel usage). The exponential formulas in $\lambda_i$ are used to induce a preference for low arrival rates whenever possible, except when the traces would become large.

For 2C scheduling in the linear KF case, we apply a similar procedure, but this time with the conditions of \cite{Liu04}:
\bea
\label{eq:linobjective}
\min_{(\lambda_1, \lambda_2) \in L \text{ s.t. \eqref{eq:lmi_lin}}}
                \tau' + e^{\frac{1}{1-\lambda_1}} + e^{\frac{1}{1-\lambda_2}},
\eea
where $\tau'=\mathrm{Trace}(V')$ with $V'$ being the solution of \eqref{eq:opt_V}.
% \subsection{Nonlinear 2C scheduling with iterative KF conditions}
% \label{subsec:nonlinSchedIter}
% As a simple applied extension of the method in \cite{Liu04} to the extended Kalman filter we iteratively apply the method by computing appropriate loop closure periods for the linearized dynamics of the system. After the loop closure period is computed from the transmission probabilities as $1/\lambda_{1,2}$ samples for a constant rate sampling. 

% Mihalis, Stefan, check the below, we can talk about any edge cases (like no measurement taken) and see which we include in the pseudocode, and which we discuss separately in the text
% To consider: do we even include the "linearize at each step" case? I think it's not adding much. Conclusion: No.
\begin{algorithm}[tb]
\caption{2C EKF scheduling with iterative KF conditions}\label{alg:iterativekf}
\begin{algorithmic}[1]
\Input candidate set $L$, dynamics $f$, threshold $\delta$ 
\State $k_1 \gets -\infty$, $k_2 \gets -\infty$
\For{each time step $k \geq 0$}
\State differentiate $f$ around $x_k$ to find $A_k$ 
\If{$k=0$, or $\Delta A_k \geq \delta$ and $k_{\mathrm{lin}} \leq \mathrm{max} \{k_1, k_2\}$} % TODO Mihalis/Stefan need to add exception for no-measurement-taken or whatever
\State solve \eqref{eq:linobjective} to obtain $(\lambda_1, \lambda_2)$
% TODO Mihalis, Stefan what if there is no feasible solution? Kinda means system is not observable, I think, if measuring at each sample is not feasible. Conclusion: we assume observability so it cannot happen.
\State $k_{\mathrm{lin}} \gets k$
\EndIf
\If{$k - k_1 \geq \lfloor\frac{1}{\lambda_1}\rfloor$}
\State measure $y_{1,k}$
\State $k_1 \gets k$
\EndIf
\If{$k - k_2 \geq \lfloor\frac{1}{\lambda_2}\rfloor$}
\State measure $y_{2,k}$
\State $k_2 \gets k$
\EndIf
\State run EKF prediction
\If{any measurement taken}
\State run EKF update by applying \eqref{eq:state_cov_update} for step $k$
\EndIf
\EndFor
\end{algorithmic}
\end{algorithm}

Since the EKF conditions may sometimes be overly conservative, we also propose an empirical alternative in which we recompute the arrival rates by solving \eqref{eq:linobjective}, whenever the linearized dynamics deviate by more than a threshold $\delta$ compared to the steps when the rates were previously computed. The deviation is measured through the 2-norm of the difference between the state transition matrix of the current linearized dynamics and that of the linearized dynamics at the time of the last recomputation: $\Delta A_k := \Vert A_k - A_{k_\mathrm{lin}}\Vert$. The extra condition $k_{\mathrm{lin}} \leq \mathrm{max} \{k_1, k_2\}$ ensures that at least one measurement was taken after the last recomputation. After each pair of rates is chosen, measurements via the two channels are read with periods $\lfloor\frac{1}{\lambda_i}\rfloor$ between recomputations.

% The recomputation involves checking each pair $(\lambda_1,\lambda_2)\in L$ for feasibility by solving the appropriate LMI, after which the bound on the estimate error covariance matrix trace is computed by solving the appropriate LMI. Knowing feasibility and a trace upper bound for each $(\lambda_1,\lambda_2)$ the optimal pair is then chosen according to \eqref{eq:l1l2objective}. 

% at each step k:
% if RELINEARIZE:
% linearize around current operating point and
% solve the problem:
% find pair of probabilities L1, L2 so that
%     (a) system is stable and
%     (b) a certain upper bound on the cost Trace(Cov) + alpha * (L1 + L2) is satisfied; or alternately this term is minimized (more computationally involved)
% endif
% read sensor i if Ti steps elapsed since the last measurement, where Ti corresponds to Li
% apply EKF

\section{Simulation Results}
\label{sec:sim}

% In this section the simulation scenarios and results obtained for each one are presented.

All the simulations use either a KF or EKF to estimate the states. The models used are discrete-time with sampling period $T_s = \qty{0.05}{\second}$, and have the form 
\bea
\label{eq:genModel}
\vx_{k+1} &= f(\vx_k) + w_k, \\
\vy_k &= \begin{bmatrix}
    C_1\\C_2
\end{bmatrix}\vx + v_k,
\eea 
with  $w$ and $v$ representing Gaussian noises with covariances \num{1e-4}$I$ and \num{1e-2}$I$ respectively. The simulations are carried out for a duration of \qty{10}{\minute}.

%For stochastic-2C estimation, the measurement channels transmit randomly at each sampling step based on the $\lambda_i$ probabilities. For 2-C scheduling, the measurements are performed every $\lfloor\frac{1}{\lambda_i}\rfloor$ steps respectively for each channel.
% \FloatBarrier
\subsection{2C Estimation and Scheduling for a Linear System}

This scenario numerically validates the method in \cite{Liu04}, as well as our proposed 2C application of that method. The dynamics model linear one-dimensional motion. Following the structure in \eqref{eq:genModel}, $f(\vx)= A\vx$ with $A=\begin{bmatrix}
    1 & 0.05 \\
    0 & 0.995 \\
\end{bmatrix}$, and $
    C_1 = \begin{bmatrix}
        1 & 0
    \end{bmatrix},
    C_2= \begin{bmatrix}
         0 & 1
    \end{bmatrix}$. The states $\vx$ include the one-dimensional position $x$ and velocity $v_x$.
%with different measurement probabilities $\lambda_i$ on each channel. The estimation error covariance matrix trace is compared to the theoretical one computed as per \cite{Liu04}.

The arrival rates vary on a grid $(\lambda_1, \lambda_2) \in L:=\{0, 0.1, 0.2, 0.3, 0.4, 0.5, 0.6, 0.7, 0.8, 0.9, 1\}^2$. Figure \ref{fig:linearTraceGrid} compares the analytical upper bounds $\tau'$ on the trace with the traces obtained in simulation, for all values of $\lambda_i$ on the grid. The analytical values indeed serve as useful upper bounds for the actual traces: they are always larger, but close in value.
\begin{figure}[!t]
    \centering
    \includegraphics[width=0.45\textwidth]{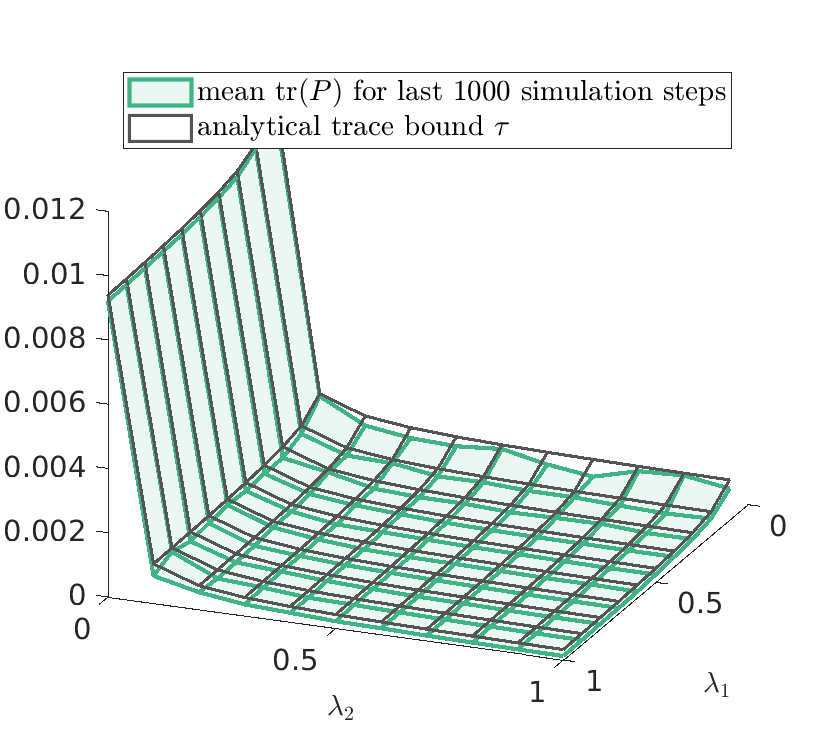}
    \caption{Analytical bound versus actual trace of the error covariance matrix in simulation.}%\textbf{[This is good :). It would be good to do it for the experimental setup as well (in the future if not now owing to the tight deadline.)]}}
    \label{fig:linearTraceGrid}
\end{figure}

Next, to solve 2C scheduling, we use \eqref{eq:linobjective} to select a pair $(\lambda_1,\lambda_2)$ from the grid. This pair turns out to be $\lambda_1 = 0.1, \lambda_2 = 0$. Note that because the system is observable from the first channel (position), the method chooses to not use the second channel at all. KF results using scheduling with this pair are shown in Figure \ref{fig:linearSchedulingTraj}. The analytical trace bound is \num{113e-4}, whereas the trace in simulation is \num{101e-4}.
\begin{figure}[!thb]
    \centering
    \includegraphics[width=0.9\linewidth]{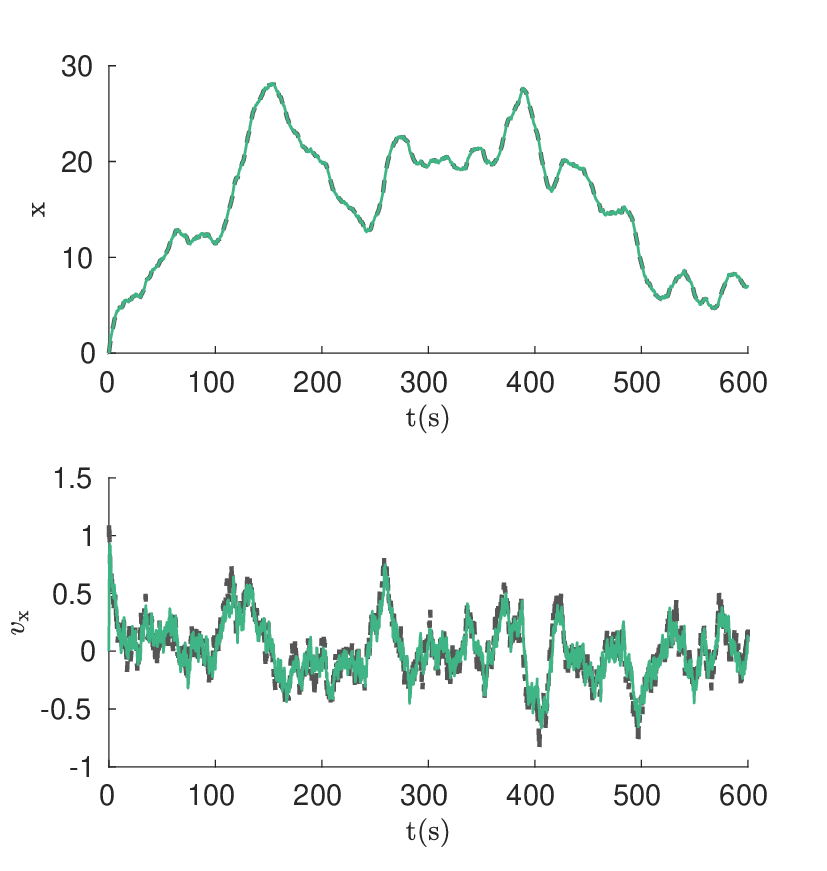}
    \caption{Real versus KF-estimated trajectories for linear model with sensor scheduling. The gray dashed lines represent the states and the teal lines are the  estimates.}
    \label{fig:linearSchedulingTraj}
\end{figure}
\begin{figure}[!thb]
    \centering
    \includegraphics[width=0.43\textwidth]{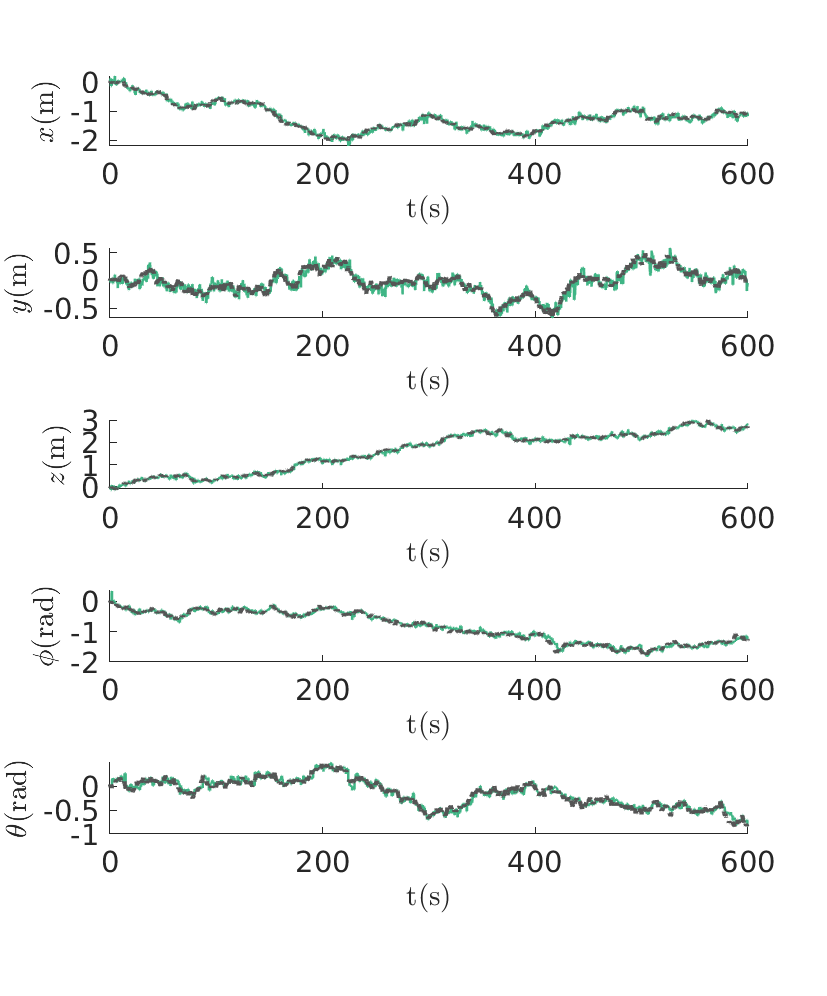}
    \vskip -1em
    \caption{Linear and angular positions and their estimates for non-linear model with sensor scheduling. The gray dashed lines represent the states and the teal lines are the  estimates.}
    \label{fig:nonlinSchedulingTraj}
    \vskip -1em
\end{figure}

% \FloatBarrier

\subsection{2C Scheduling for Nonlinear Dynamics}
\label{sub:nonlinsim}
In the nonlinear case, the dynamics used comprise a five degree-of-freedom (5-DOF) constant acceleration kinematic model. The degree of freedom removed from the standard 6-DOF model is the pitch angle. This was done partly to control computational complexity and partly because the real UUV for which we will later use these dynamics lacks this degree of freedom. %The reason for choosing this model is its usage in the Robot Operating System (ROS) EKF package,
Keeping the structure of \eqref{eq:genModel}, the state vector is $\vx =[x, y, z, \phi, \psi, v_x, v_y, v_z, v_\phi, v_\psi, a_x, a_y, a_z]^T$, and 
\bea
&f(\vx) =\\&\vx + T_s\begin{bmatrix}
     v_x\cos\psi -v_y\sin\psi\cos\phi + v_z\sin\psi\sin\phi\\
    v_x\cos\psi-v_y\sin\psi\cos\phi + v_z\sin\psi\sin\phi\\
    v_y\sin\phi + v_z\cos\phi\\
    v_\phi\\
    v_\psi\cos\phi\\
    a_x\\
    a_y\\
    a_z\\
    0_{5\times1}
\end{bmatrix},\\
&C_1 = \begin{bmatrix}
 %    0_{3\times3} & 0_{3\times7}\\
     0_{7\times3} & 1_{7\times7}
 \end{bmatrix},\, C_2 = \begin{bmatrix}
     1_{3\times3} & 0_{3\times7}\\
  %   0_{7\times3} & 0_{7\times7}
 \end{bmatrix}.
 \eea  

%compared to the theoretical one computed as per the method described in this paper.

% Figure \ref{fig:nonlinTraceGrid} compares the theoretical estimation error covariance matrix trace with the one obtained through simulation for all pairs $(\lambda_1,\lambda_2)$ on the grid. The missing values on the grid correspond to infeasible $\lambda_1,2$ pairs. 
% \begin{figure}[!htb]
%     \centering
%     \includegraphics[width=0.47\textwidth]{figs/nonlinTraceGrid.eps}
%     \caption{Theoretical vs simulation error covariance matrix trace compari-
% son for EKF 2C}
%     \label{fig:nonlinTraceGrid}
% \end{figure}
% \FloatBarrier
The 2C scheduling problem \eqref{eq:l1l2objective} is solved with $(\lambda_1,\lambda_2) \in L:=\{0.001, 0.01, 0.1, 0.5, 0.625\}^2$. The chosen pair is $\lambda_1 = 0.1, \lambda_2 = 0.1$. For this pair, the trajectories of the linear and angular positions and of their respective estimates are shown in Figure \ref{fig:nonlinSchedulingTraj}. The analytical trace bound is \num{0.3874}, whereas the trace in simulation is \num{0.0625}, showing that the conditions are conservative.

This motivates us to perform another simulation, this time using Algorithm \ref{alg:iterativekf} to compute the feedback periods for the two channels. The threshold for the deviation of the dynamics of the system is  $\delta = 0.1$, chosen experimentally. The state trajectories are very similar to Figure \ref{fig:nonlinSchedulingTraj}, so they are not shown here. Instead, Figure \ref{fig:nonlinIterativeFBPeriod} shows the feedback periods chosen along the simulation for both channels, along with the feedback period which corresponds to the single pair $\lambda_1 = 0.1, \lambda_2 = 0.1$ chosen using \eqref{eq:l1l2objective}. It is apparent that Algorithm \ref{alg:iterativekf} produces less conservative periods, albeit without analytical guarantees. The trace for this simulation is \num{0.1316}. 
\begin{figure}[!tb]
    \centering
    \includegraphics[width=0.55\linewidth]{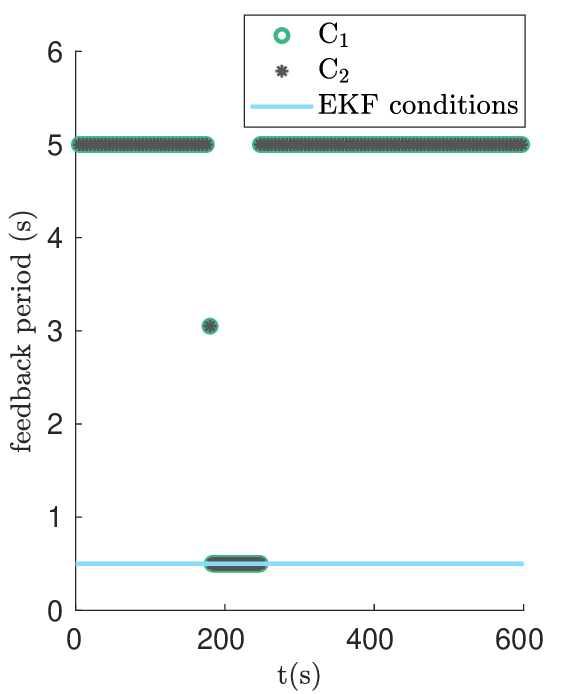}
    \caption{Feedback periods for 2C scheduling with the non-linear model.}
    \label{fig:nonlinIterativeFBPeriod}
\end{figure}

\section{Experimental Results for an Underwater Robot}
\label{sec:exp}

Finally, we apply 2C scheduling with real data collected for UUV pose estimation in the SeaClear project \cite{seaclear}, see Figure \ref{fig:rovImg}. The two channels correspond to the internal and the off-board sensors. Similarly to the non-linear problem of Section \ref{sub:nonlinsim}, the UUV measures its angular position and velocity using an inertial measurement unit, and its linear velocity using a Doppler velocity logger. Differently from the simulated problem, the UUV also has direct access to the measurement of its depth $z$ using a pressure sensor. As a result, the second channel only communicates the position of the UUV in the XY-plane, determined from camera images of a UAV, while the $z$ coordinate is communicated on the first channel.
\begin{figure}[!thb]
        \centering
        \includegraphics[width=0.23\textwidth]{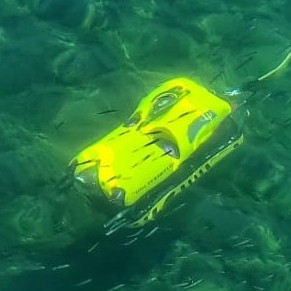}\ 
        \includegraphics[width=0.23\textwidth]{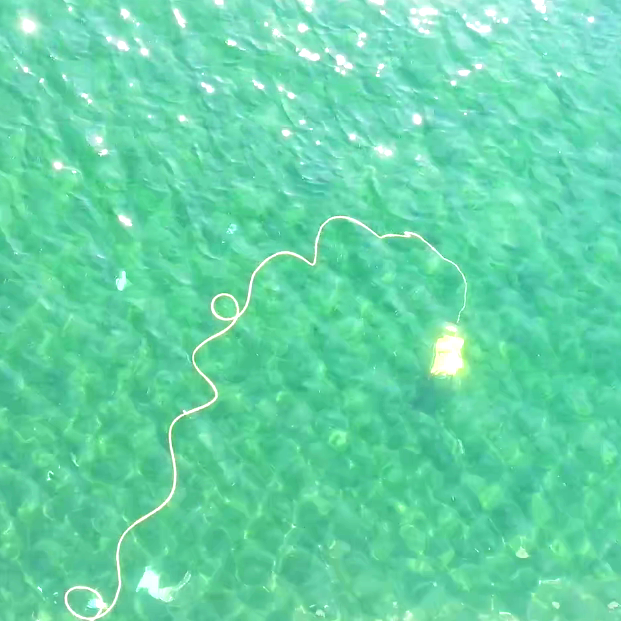}
        \caption{Left:  MiniTortuga, the UUV used for the experiment. Right: Example camera image of the UAV from the collected data set, used to determine the position of the UUV (the image has been cropped for better visibility).}
        \label{fig:rovImg}
\end{figure}

The internal sensors provide feedback at each timestep $(\lambda_1=1)$ and only the loop closure period for the data sent by the UAV varies. The rate $\lambda_2$ is selected from the same set of discrete values as in Section \ref{sub:nonlinsim}, $\lambda_2 \in \{0.001, 0.01, 0.1, 0.5, 0.625\}$. The linearized model used is identical to that from the simulations. The EKF is implemented using the robot\_localization library \cite{Moore2016}.

In this experiment, it turns out that irrespective of whether \eqref{eq:l1l2objective} or the iterative Algorithm \ref{alg:iterativekf} is applied, the resulting value of $\lambda_2$ is always the same: $0.01$. Thus, in this case the EKF solution is less conservative than in the simulations, possibly because along this trajectory the angles only rotate the dynamics, without significantly affecting the stability properties of the system. % As a result, the feedback period required to keep the estimation accurate does not change when the feedback period is recomputed for different operating points along the sampled trajectory. 

Figure \ref{fig:rovXY} compares the positions in the plane estimated using the EKF with channel scheduling, versus positions read from a short baseline (SBL) acoustic positioning system. The SBL measurement is used as a proxy for the ground truth position of the UUV. It can be seen that the estimated position is close to the one measured by the SBL, with a root mean squared error of $\qty{1.3915}{\metre}$ between the two.
\begin{figure}[!tb]
    \centering 
    \includegraphics[width=0.47\textwidth]{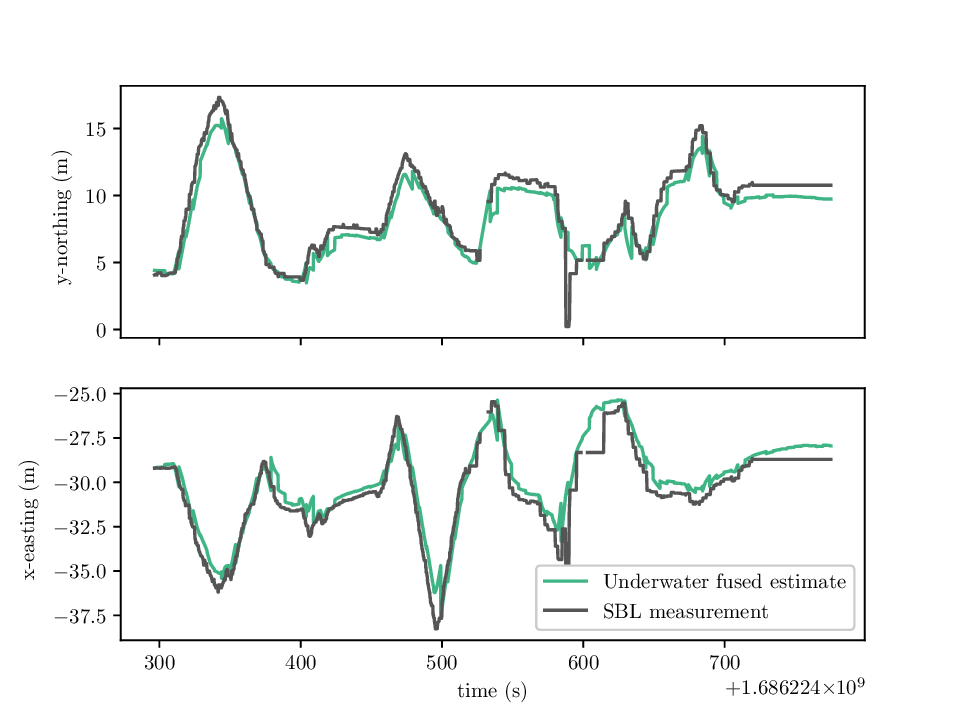}
    \caption{Estimation results using channel scheduling for the real UUV. The horizontal axis displays Unix epoch time.}
    \label{fig:rovXY}
\end{figure}

\section{Conclusions}
\label{sec:conclusions}

This paper characterized the estimation error for an EKF that reads sensors over two channels that drop measurements stochastically, and proposed a solution to deterministically choose when to read the channels when they are under the control of the estimator. The approaches were validated both in simulations and on real data. To generalize the approach in future work, several interesting directions emerge: allowing for an arbitrary number of channels instead of just two, deriving similar results for the unscented Kalman filter, or taking into account specific scheduling protocols.

% Future work:
% -- multiple channels accounted for
% -- specific scheduling protocols taken into account
% -- derive similar results for Unscented Kalman Filter (UKF) and compare
%-- conditions on Jacobian

\bibliographystyle{ieeetr}
\bibliography{references}
\vspace{12pt}
\end{document}